\documentclass[11pt]{article}
\usepackage[utf8]{inputenc}
\usepackage[finnish,english]{babel}					
\usepackage{amscd,amssymb,amsfonts}
\usepackage[utf8]{inputenc} 
\usepackage[T1]{fontenc}    
\usepackage{hyperref}       
\usepackage{url}            
\usepackage{booktabs}       
\usepackage{nicefrac}       
\usepackage{microtype}      
\usepackage{lipsum}
\usepackage{amsmath}
\usepackage{graphicx}
\usepackage{color}
\usepackage{hyperref}
\usepackage{subcaption}
\usepackage{caption}
\usepackage{multirow}
\usepackage{appendix}
\usepackage{pdfpages}
\usepackage{float}
\usepackage[ruled,vlined]{algorithm2e}
\usepackage{amsthm}
\usepackage{authblk}
\usepackage{comment}

\newcommand{\argmin}{\mathop{\mbox{arg\,min}}}
\newcommand{\gvec}{\mathbf{g}}
\newcommand{\R}{\mathbb{R}}

\newcommand{\vvec}{\mathbf{v}}
\newcommand{\mvec}{\mathbf{m}}
\newcommand{\A}{\mathcal{A}}

\newcommand{\Regul}{\mathcal{R}}
\newcommand{\Sep}{\mathcal{S}}
\newcommand{\Radon}{A}

\newtheorem{theorem}{Theorem}[section]

\newtheorem{lemma}[theorem]{Lemma}

\title{Inner product regularized multi-energy X-ray tomography for material decomposition}
\author[1]{Salla-Maaria Latva-Äijö}
\author[2]{Filippo Zanetti}
\author[3]{Ari-Pekka Honkanen}
\author[3]{Simo Huotari}
\author[2]{Jacek Gondzio}
\author[1]{Matti Lassas}
\author[1]{Samuli Siltanen}

\affil[1]{Department of Mathematics and Statistics, University of Helsinki, Helsinki, Finland}
\affil[2]{School of Mathematics, University of Edinburgh, Scotland, UK}
\affil[3]{Department of Physics, University of Helsinki, Helsinki, Finland}

\date{June 2023}

\begin{document}

\maketitle

\section*{Abstract}

Multi-energy X-ray tomography is studied for decomposing three materials using three X-ray  energies and a classical energy-integrating detector. A novel regularization term comprises inner products between the material distribution functions, penalizing any overlap of different materials. The method is tested on real data measured of a phantom embedded with  Na$_2$SeO$_3$,  Na$_2$SeO$_4$, and elemental selenium. It is found that the two-dimensional distributions of selenium in different oxidation states can be mapped and distinguished from each other with the new algorithm. The results have  applications in material science, chemistry, biology and medicine. 


\newpage

\tableofcontents

\section{Introduction}

We study multi-energy X-ray tomography for material decomposition, generalizing the reconstruction method we introduced in \cite{gondzio2022material}. There we demonstrated with simulated data the feasibility of one inner product regularizer term in the context of dual-energy X-ray tomography for a two-material object. Here we work with three different X-ray energies and materials, applying our generalized reconstruction method to measured monochromatic X-ray data. We note that our methods assume classical energy-integrating sensors and do not need photon counting detectors.

The new regularizer takes the form 
\begin{equation}\label{regu_intro}
    \langle \gvec^{(1)},\gvec^{(2)}\rangle + \langle \gvec^{(1)},\gvec^{(3)}\rangle + \langle \gvec^{(2)}, \gvec^{(3)}\rangle,
\end{equation}
where the non-negative vector $\gvec^{(j)}$ models the distribution of material $j$ inside the target. Now the penalty (\ref{regu_intro}) promotes the {\it a priori} information that the target consists of just one material at any point, since (\ref{regu_intro}) becomes zero if at least one function in each inner product vanishes at a pixel. 

Clinical computed tomography (CT) devices are usually built with conventional X-ray tubes, which produce polychromatic broad-bandwidth beams. Such beams lead to beam-hardening artifacts and insensitivity to the chemical composition of the imaged object. 
These limitations can be overcome with  monochromatic X-ray beams. They also allow taking advantage of sharp changes in attenuation as function of energy (K-edge variations) of certain materials. 

Monochromatic data was produced with a new method, where a polychromatic X-ray beam is directed through a spherically bent crystal analyser \cite{honkanen2023monochromatic}. It monochromatises and refocuses the beam at the Rowland circle which acts as a secondary source. This makes it possible to adjust the energy of the diffracted photons and record the spatially resolved changes in the attenuation coefficients. All of this can be done easily in a traditional X-ray laboratory (as opposed to a synchrotron facility). See \cite{honkanen2023monochromatic}. 

Previous literature on the topic is mostly based on the discrete tomography approach \cite{batenburg20093d, batenburg2011dart, baumann2007discrete, herman2008advances, herman2012discrete}. Our algorithm is different from them as it is based on the regularizing penalty based on inner products.
We remark that using an energy-sensitive  photon counting detector opens up further possibilities, see e.g. \cite{long2014multi,wu2020dictionary,solem2021material}. We do not discuss such avenues of investigation further in this paper. 

Our reconstruction results suggest that the inner product regularizer is promising for material decomposition. There are potential applications in chemistry, biology and medicine.

\section{Measurement model}

We restrict here to the intersection of a physical body with a two-dimensional square $\Omega\subset\R^2$. The measured X-rays are thus assumed to travel in the plane determined by $\Omega$. This restriction is only for the simplicity of the exposition and computation; the method generalizes also to higher dimensions.

We discretize $\Omega$  into $N\times N$ square-shaped pixels. 
There are three unknowns: non-negative $N\times N$ matrices  $G^{(1)}$, $G^{(2)}$ and $G^{(3)}$ modelling the distributions of material 1, material 2 and material 3, respectively. The number $G^{(\ell)}_{i,j}\geq 0$ represents the concentration of material $\ell$ in pixel $(i,j)$, where $i$ is the row index and $j$ is the column index. In numerical computations, we represent the elements of the material matrices $(G^{(1)},G^{(2)},G^{(3)})\in(\R^{N\times N})^3$, as a vertical vector 
$$
\gvec = \left[\!\!\begin{array}{l}\gvec^{(1)}\\\gvec^{(2)}\\\gvec^{(3)}\end{array}\!\!\right]\in \R^{3N^2},
$$
where the vertical vector $\gvec^{(j)}\in\R^{N^2}$ is formed by stacking the columns of the matrix $G^{(j)}\in \R^{N\times N}$.

We record X-ray transmission data with three different energies, calling them low, middle and high energy. Measurements are resulting in three $M$-dimensional data vectors called $ \mvec^L$, $ \mvec^M$ and $\mvec^H$. For example, the low-energy measurement is given by
\begin{equation}\label{measmodel_low}
    \mvec^L = c_{11}\Radon^L\gvec^{(1)} + c_{12}\Radon^L\gvec^{(2)} + c_{13}\Radon^L\gvec^{(2)},
\end{equation}
and similar form applies to all the different energies.

The materials attenuate the X-rays with individual strengths described by the energy dependent attenuation constants $c_{11}>0$, $c_{12}>0$, $c_{13}>0$, ... and so on. The empirical values of $c_{ij}$ can be determined by measuring pure samples of each of three known materials. The $D{\times}N^2$ matrix $A^L$ encodes the geometry of the tomographic measurement using the standard pencil-beam model. Here $D$ is the number of sensors (pixels) in the linear X-ray detector. The geometric system matrices $\Radon^M$, $\Radon^H$ and $\Radon^L$ can be different from each other. We have kept the different letters for all the system matrices in the theory part, but in the actual calculations, we used in the same measurement geometry and system matrix for all the measurements: $\Radon^H=\Radon^M=\Radon^L$.

We can combine all the measurements in a unified linear system:
\begin{equation}\label{unifiedsystem}
    \mvec=\left[\!\!\begin{array}{l}\mvec^L\\\mvec^M\\\mvec^H\end{array}\!\!\right]
    =\begin{bmatrix}
c_{11}\Radon^L &c_{12}\Radon^L &c_{13}\Radon^L\\
c_{21}\Radon^M &c_{22}\Radon^M &c_{23}\Radon^M\\
c_{31}\Radon^H &c_{32}\Radon^H &c_{33}\Radon^H\\
\end{bmatrix}
\left[\!\!\begin{array}{l}\gvec^{(1)}\\\gvec^{(2)}\\\gvec^{(3)}\end{array}\!\!\right]=\A \gvec.
\end{equation}
The core idea in this triple-energy X-ray tomography for material decomposition is to choose the three energies so that three materials respond to them differently. The solution of (\ref{unifiedsystem}) is rather analogous to solving a system of three linear equations for three variables.

We use our novel variational regularization approach in the space $\gvec\in \R^{3N^2}$, including a non-negativity constraint:
\begin{equation}\label{variational_cont}
\widetilde{\gvec}_{\alpha,\beta} 
= \argmin_{\gvec^{(j)}\geq 0}
\left\{\|\mvec-\A \gvec\|_2^2 + \alpha\Regul(\gvec) + \beta\Sep(\gvec) \right\}, 
\end{equation}
where $\alpha,\beta>0$ are regularization parameters, $\gvec^{(j)}\geq 0$ means that the elements of the vector are non-negative numbers and the regularizer $\Regul$ can be any of the standard choices such as the Tikhonov penalty
\begin{equation}\label{tikhonov_penalty}
    \Regul(\gvec) = \|\gvec\|_2^2. 
\end{equation}
%
%
The novelty arises from the term that penalises all Inner Products (IP) between vectors $\gvec^{(1)}$, $\gvec^{(2)}$ and  $\gvec^{(3)} \in \R^{N^2}$. The regularization term $\Sep(g)$ is defined as double the sum of those:
\begin{equation}\label{materialsep_penalty}
\begin{split}
\Sep(\gvec) = \Sep(\left[\!\!\begin{array}{l}\gvec^{(1)}\\\gvec^{(2)}\\\gvec^{(3)}\end{array}\!\!\right]) := 2(\langle \gvec^{(1)},\gvec^{(2)}\rangle + \langle \gvec^{(1)},\gvec^{(3)}\rangle + \langle \gvec^{(2)}, \gvec^{(3)}\rangle)\\
= 2\sum_{i=1}^{N^2} \gvec^{(1)}_i \gvec^{(2)}_i + 2\sum_{i=1}^{N^2}  \gvec^{(1)}_i \gvec^{(3)}_i + 2\sum_{i=1}^{N^2}  \gvec^{(2)}_i \gvec^{(3)}_i.
\end{split}
\end{equation}
Together with the non-negativity constraint, $\Sep$ promotes the point-wise separation between three materials: {\it at each pixel, at least two of the images, $G^{(1)}$, $G^{(2)}$ or  $G^{(3)}$, needs to have a zero value to make $\Sep$ minimal.}

The quadratic program resulting from the application of the novel variational regularization is solved using an Interior Point Method (IPM) \cite{wright,gondzio_25}; we develop an efficient preconditioner for the normal equations which guarantees the spectrum of the preconditioned matrix to remain independent of the IPM iteration. 

\section{Preconditioned interior point method}\label{sec:IPM}

In this section we discuss optimization with preconditioned interior point method (IPM).

By combining the use of Tikhonov regularizer (\ref{tikhonov_penalty}) and the Inner Product regularizer (\ref{materialsep_penalty}), which promotes the point-wise separation of three materials, we arrive at the constrained quadratic programming task
\begin{equation}\label{variational_Tikhonov}
\argmin_{\gvec^{(j)}\geq 0}
\left\{\|\mvec-\A \gvec\|_2^2 + \alpha\|\gvec\|_2^2 + \beta\, \gvec^T L \gvec \right\},
\end{equation}
where 
$$
  L=
  \begin{bmatrix}
      0 & I & I \\
      I & 0 & I \\
      I & I & 0 \\
 \end{bmatrix},
$$

with nine blocks of size $N^2{\times}N^2$ each. In general, using the regularizer with $k$ materials is equivalent to adding a regularization term $\gvec^T L\gvec$, with
\[
L_k= E_k\otimes I_{N^2}-I_{kN^2}
\]
where $E_k$ is the matrix of all ones of size $k$, $E_k=[\mathbf e\, \mathbf e\, \dots \mathbf e]$

The problem may be written as an explicit quadratic program with inequality (non-negativity) constraints

\begin{equation}
    \argmin_{\gvec^{(j)}\geq 0} -\mvec^T \A \gvec +\frac{1}{2} \gvec^T(Q_1+Q_2)\gvec
    \label{minproblem}
\end{equation}








where


\begin{align}
\label{Qkronecker}
Q_1&=
\begin{bmatrix}
c_{11}^2 & c_{11}c_{12} & c_{11}c_{13}\\
c_{11}c_{12} & c^2_{12} & c_{11}c_{13}\\
c_{11}c_{13} & c_{12}c_{13} & c^2_{13}
\end{bmatrix}\otimes (A^L)^TA^L+\\
&+\begin{bmatrix}
c_{21}^2 & c_{21}c_{22} & c_{21}c_{23}\\
c_{21}c_{22} & c^2_{22} & c_{22}c_{23}\\
c_{21}c_{23} & c_{22}c_{23} & c^2_{23}
\end{bmatrix}\otimes (A^M)^TA^M+ \notag\\
&+\begin{bmatrix}
c_{31}^2 & c_{31}c_{32} & c_{31}c_{33}\\
c_{31}c_{32} & c^2_{32} & c_{32}c_{33}\\
c_{31}c_{33} & c_{32}c_{33} & c^2_{33}
\end{bmatrix}\otimes (A^H)^TA^H,\notag
\end{align}
\begin{equation}\label{Q_regularize}
Q_2 = \begin{bmatrix} 
\alpha & \beta & \beta\\ 
\beta & \alpha & \beta\\
\beta & \beta & \alpha\end{bmatrix}\otimes I_{N^2},
\end{equation}
where $\otimes$ represents the Kronecker product.

For an arbitrary number of materials $k$, $Q_1$ is still equal to $\mathcal A^T\mathcal A$ and $Q_2$ is equal to
\begin{equation}\label{Q2_k_materials}
Q_2 = \beta L_k + \alpha I_{kN^2} = 
\underbrace{\begin{bmatrix}
\alpha & \beta & \beta & \dots & \beta\\
\beta & \alpha & \beta & \dots & \beta\\
\beta & \beta & \alpha & \ddots & \vdots\\
\vdots & \vdots & \ddots & \ddots & \beta\\
\beta & \beta & \dots & \beta & \alpha
\end{bmatrix}}_{Q'_2\in\mathbb R^{k\times k}}
\otimes I_{N^2}
\end{equation}

Recall this important property about the Kronecker product:
\begin{lemma}
\label{kroneig}
Given two square matrices $T$ and $Z$, the eigenvalues of the Kronecker product $T \otimes Z$ are given by $t \cdot z$, where $t$ is an eigenvalue of $T$ and $z$ is an eigenvalue of $Z$.
\end{lemma}



To study the convexity of problem \eqref{minproblem}, let us first analyze the eigenvalues of matrix $Q_2'$.
\begin{lemma}
\label{Q2eigenvalues}
When considering $k$ materials, matrix $Q_2'$ has eigenvalues $\alpha-\beta$ with multiplicity $k-1$ and $\alpha+(k-1)\beta$ with multiplicity $1$.
\end{lemma}
\begin{proof}
The result follows from noticing that
\[Q_2'-(\alpha-\beta)I_k = \begin{bmatrix}
\beta & \dots & \beta\\
\vdots & \ddots & \vdots\\
\beta & \dots & \beta
\end{bmatrix}
\]
has rank $1$ and 
\[
Q_2'-(\alpha+(k-1)\beta)I_k = \begin{bmatrix}
-(k-1)\beta & \beta & \dots & \beta\\
\beta & -(k-1)\beta & \beta & \vdots\\
\vdots & & \ddots & \\
\beta & \dots & \beta & -(k-1)\beta
\end{bmatrix}
\]
has rank $k-1$.

\end{proof}

We can now give a sufficient condition for the convexity of problem \eqref{minproblem}.
\begin{lemma}
\label{lemmaconvex}
If $\alpha\ge\beta$, problem \eqref{minproblem} is convex, for any number of materials $k$.
\end{lemma}
\begin{proof}
If $\alpha\ge\beta$, then matrix $Q_2'$ has only positive eigenvalues, due to Lemma~\ref{Q2eigenvalues}. Therefore, matrix $Q_2$ is positive semi-definite, due to Lemma~\ref{kroneig}. Since $\mathcal A^T\mathcal A$ is always positive semi-definite, $Q = \mathcal A^T\mathcal A+Q_2$ is positive semi-definite if $Q_2$ is. Therefore, the problem is a convex quadratic program.





\end{proof}
Therefore, in the following we assume that $\alpha\ge\beta$.

We now generalize the eigenvalue bounds for the preconditioned matrix found in \cite{gondzio2022material} to the case of $k$ materials. We recall that the matrix of the Newton system that arises at each interior point iterations can be written as $M=Q_1+Q_2+\Theta$, where $Q_1$ is given in \eqref{Qkronecker}, $Q_2$ is given in \eqref{Q_regularize} and $\Theta$ is a diagonal matrix with strictly positive entries. We restrict to the case of imaging using the same angles for all energies (i.e.\ $A^L=A^M=A^H=:A$); then, in the case of $k$ materials, $Q_1$ can be written as
\[
Q_1 = (\mathcal C^T\mathcal C) \otimes (A^TA),\qquad \mathcal C=\begin{bmatrix}
    c_{11} & c_{12} & \dots  & c_{1k}\\
    c_{21} & c_{22} & \dots  & c_{2k}\\
    \vdots &        & \ddots &       \\
    c_{k1} & c_{k2} & \dots  & c_{kk}
\end{bmatrix}.
\]
We suppose to precondition matrix $M$ with $P=(\mathcal C^T\mathcal C) \otimes(\rho I_{N^2})+Q_2+\Theta$, where $\rho\in\mathbb R$ is an approximation of the mean diagonal element of $A^TA$. Notice that this matrix has a $k\times k$ block structure with diagonal blocks and is thus easy to invert and apply to a vector. 

Then, the following eigenvalue bounds hold:
\begin{lemma}
    Given $M=(\mathcal C^T\mathcal C) \otimes (A^TA) + Q2'\otimes I_{N^2} + \Theta$ and $P=(\mathcal C^T\mathcal C) \otimes (\rho I_{N^2}) + Q2'\otimes I_{N^2} + \Theta$, the eigenvalues of $P^{-1}M$ are real and lie in the interval
    \[
    \Big[
    \frac{\alpha-\beta}{\rho\sigma_\text{max}^2(\mathcal C)+\alpha+(k-1)\beta},
    \frac{\sigma_\text{max}^2(A)\sigma_\text{max}^2(\mathcal C)+\alpha+(k-1)\beta}{\rho\sigma_\text{min}^2(\mathcal C)+\alpha-\beta}
    \Big].
    \]
\end{lemma}
\begin{proof}
    For a certain eigenvector $\vvec$ of unit norm, the eigenvalues of $P^{-1}M$ can be written as
    \[
    \lambda = \frac{\vvec^TM\vvec}{\vvec^TP\vvec} = \frac{\vvec^T\big((\mathcal C^T\mathcal C) \otimes (A^TA)\big)\vvec + \vvec^T\big(Q_2'\otimes I_{N^2}\big)\vvec + \vvec^T\Theta\vvec}{\vvec^T\big((\mathcal C^T\mathcal C) \otimes (\rho I_{N^2})\big)\vvec + \vvec^T\big(Q_2'\otimes I_{N^2}\big)\vvec + \vvec^T\Theta\vvec}.
    \]
Using basic properties of the Kronecker product (see \cite{gondzio2022material} for more details), notice that
\[
\vvec^T\big((\mathcal C^T\mathcal C) \otimes (A^TA)\big)\vvec \,\in\, \big[0,\sigma_\text{max}^2(A)\sigma_\text{max}^2(\mathcal C)\big],
\]
\[
\vvec^T\big((\mathcal C^T\mathcal C) \otimes (\rho I_{N^2})\big)\vvec \,\in\, \big[\rho\sigma_\text{min}^2(\mathcal C),\rho\sigma_\text{max}^2(\mathcal C)\big],
\]
\[
\vvec^T\big(Q_2'\otimes I_{N^2}\big)\vvec \,\in\, \big[\alpha-\beta,\alpha+(k-1)\beta\big].
\]
With similar arguments as in \cite{gondzio2022material} the thesis can be easily derived.
\end{proof}
Similar bounds can be found in the case $A^L \ne A^M \ne A^H$. Notice that when $k$ becomes larger, the effect of the inner product regularization coefficient $\beta$ on the eigenvalues of $P^{-1}M$ becomes more significant with respect to $\alpha$. In particular, if $k$ grows, for the same values of $\alpha$ and $\beta$, the spectrum of $P^{-1}M$ widens, worsening the efficacy of the preconditioner. A possible solution to this issue is to increase the gap between $\alpha$ and $\beta$ when the number of materials grows.

\section{Experimental data}

We use monochromatic X-ray data described in \cite{honkanen2023monochromatic}. There, three-dimensional distributions of selenium in different oxidation states is mapped and distinguished from background and each other with absorption edge contrast tomography.

 \subsection{The phantom}\label{sec:phantom}
 
In \cite{honkanen2023monochromatic}, three holes were drilled into a cuboid acrylic glass (PMMA) phantom. The holes where filled with selenium compounds: elemental Se, Na$_2$SeO$_3$ and Na$_2$SeO$_4$ and starch. See Figure \ref{fig:XCAS_phantom} for a schematic drawing of the phantom, where we show the materials with different colors.  Figure \ref{fig:selenium_phantom_photo} shows a photograph of the sample. Note how there seems to be no visible difference between the measured materials: they all have same color and structure.
  
  \begin{figure}[H]
      \centering
      \includegraphics[scale=0.4]{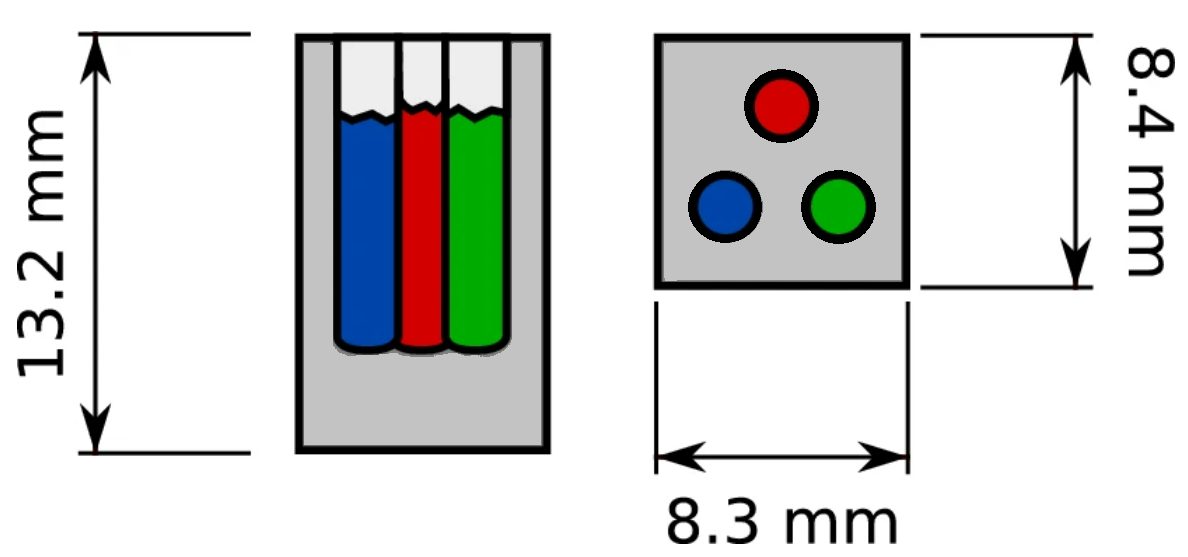}
      \caption{The figure shows a schematic drawing of the phantom. The cuboid PMMA (grey) was drilled with three holes each of which was filled with a mixture of a selenium compound [elemental Se (red), Na$_2$SeO$_3$ (green) and  Na$_2$SeO$_4$ (blue)] and starch. The holes were capped with tissue paper (very light grey).}
      \label{fig:XCAS_phantom}
  \end{figure}

\begin{figure}[H]
    \centering
    \includegraphics[scale=0.5]{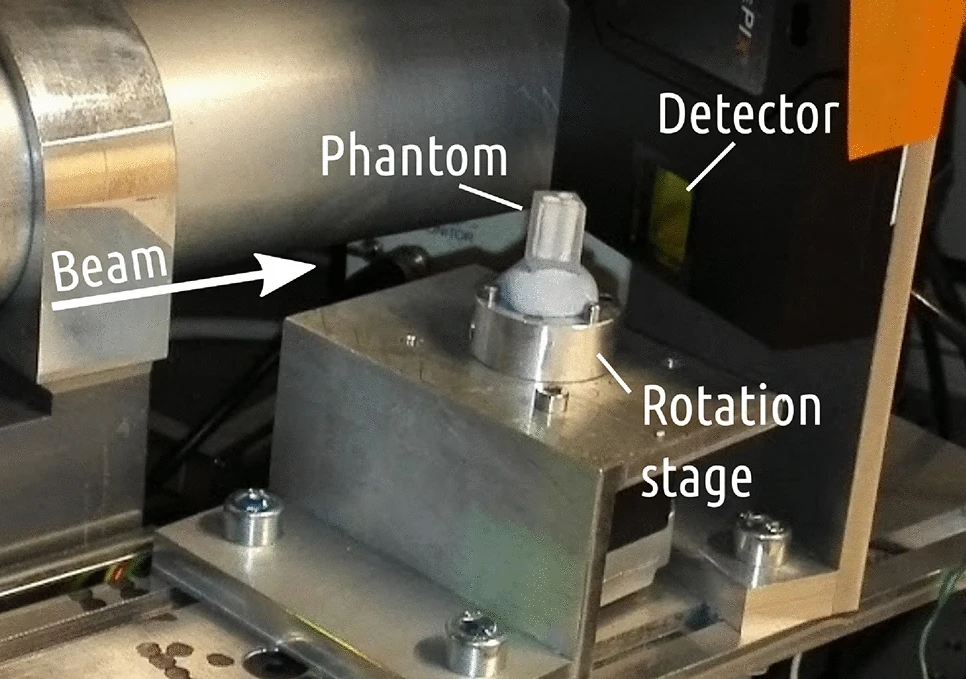}
    \caption{Here is a photograph of the tomography setup. The detector in the picture is a different TimePIX-based model than the one used in this work. There is no  visible difference between  the materials.}
    \label{fig:selenium_phantom_photo}
\end{figure}

\subsection{X-ray measurement}



The phantom described in Section \ref{sec:phantom} was measured with a XAS-CT setup, see Figure  \ref{fig:XCAS_equipment}. This kind of arrangement monochromatises the polychromatic X-rays from the X-ray tube when they travel through a spherically bent crystal analyser.
 \begin{figure}[H]
     \centering
     \includegraphics[scale=0.4]{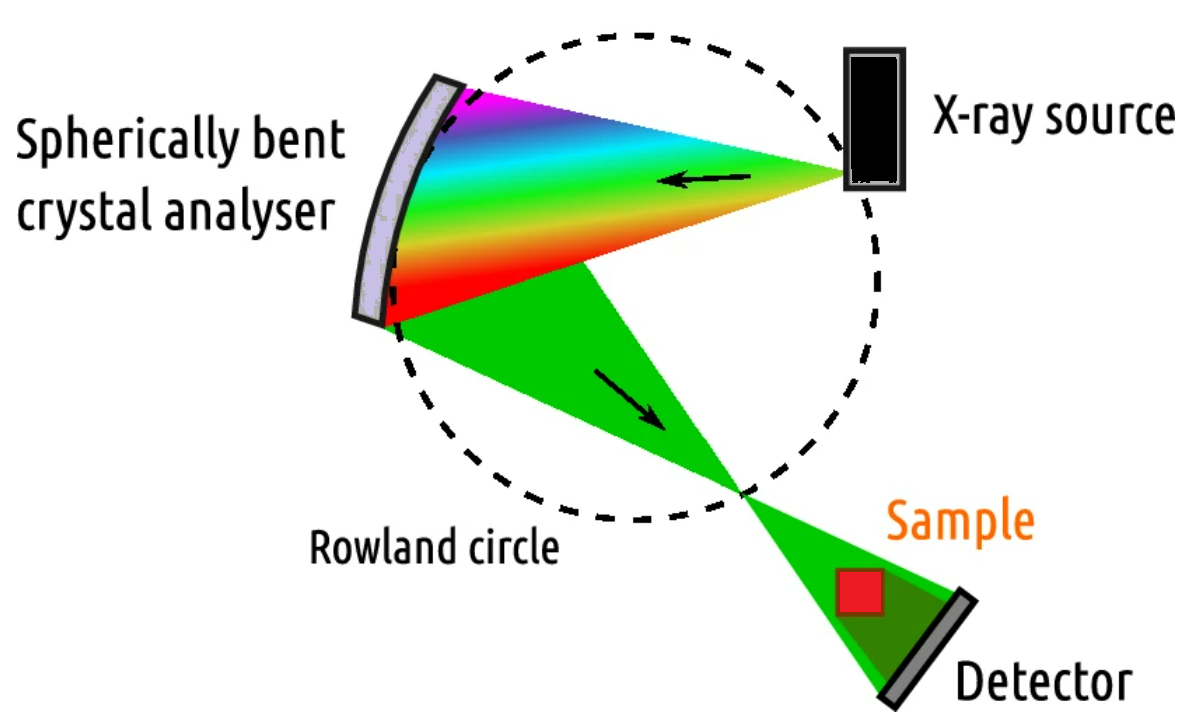}
     \caption{The schematic drawing of the XAS-CT setup. The polychromatic X-rays produced by the X-ray tube are monochromatised with the spherically bent crystal analyser. The sample to be imaged is illuminated by the monochromatised beam by moving it away from the Rowland circle so that the defocused beam completely covers it. The beam transmitted through the sample is recorded with a position-sensitive detector.\cite{honkanen2023monochromatic}}
     \label{fig:XCAS_equipment}
 \end{figure}

  



  \begin{figure}[H]
      \centering
      \includegraphics[scale=0.8]{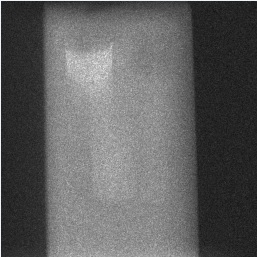}
      \includegraphics[scale=0.8]{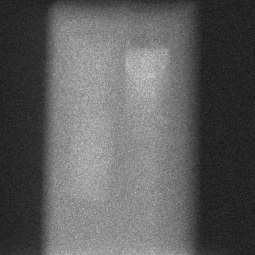}
      \caption{Two random example projection images of the sample. These projection images were taken with energies 12.658 keV and 12.662 keV.}
      \label{fig: projection_images}
  \end{figure}

  \begin{figure}[H]
      \centering
      \includegraphics[scale=0.3]{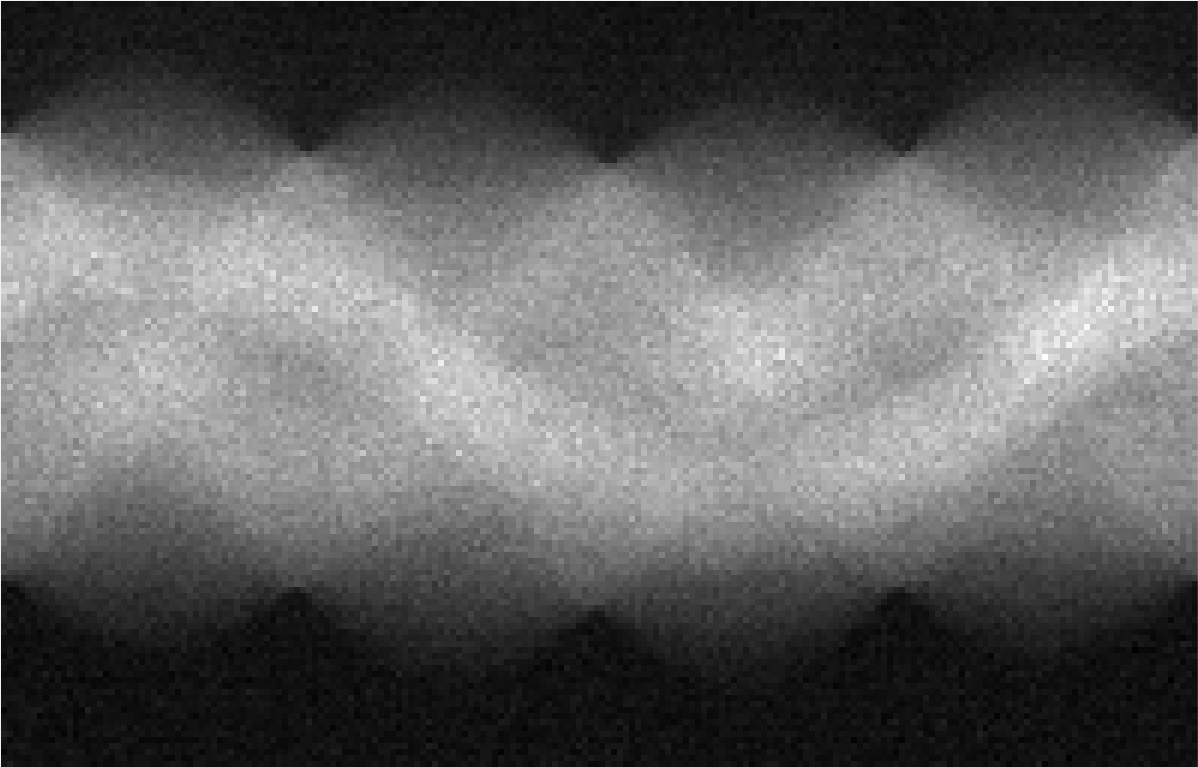}
      \caption{Example sinogram of the data. The row 80 (81 in Matlab syntax) was picked from all the projection images to the sinograms. There was all in all three different sinograms, corresponding the energies 12.658 keV, 12.662 keV and 12.685 keV. These energies correspond the low, middle and high energies in our measurement model. This sinogram is produced from 12.662 keV projection images.}
      \label{fig:example_sinogram}
  \end{figure}

\subsection{K-edges in X-ray attenuation}

Our IP method works on its best when we investigate materials with X-ray energies which are near the K-edges of the materials. In X-ray absorption spectroscopy, the K-edge is a sudden increase in X-ray absorption, occurring when the energy of the X-rays is just above the binding energy of the atoms' electron shell. The term is based on the International Union of Pure and Applied Chemistry (IUPAC) X-ray notation, where the atomic orbitals are labelled with letters K, L, M, N... and so on. The innermost electron shell is known as the K-shell. 

Physically, the sudden increase in attenuation is caused by the photoelectric absorption of photons. For this interaction to occur, the photons must have more energy than the binding energy of the K-shell electrons (K-edge).  A photon having an energy just above the binding energy of the electron is more likely to be absorbed than a photon having an energy just below the binding energy or significantly above it. 

Figure \ref{fig:Se_att_coeffs} illustrates this. In the figure, we have the different attenuation values $\frac{\mu}{\rho}$ as a function of photon energies in keV:s. The plots show how the attenuation varies with three different materials, which are pure selenium Se (blue), Na$_2$SeO$_3$ (red) and PMMA (green). We can see the sudden increase in attenuation of the selenium containing samples at the measured X-ray energy range (black box). We can compare with the PMMAs attenuation values, which are very smooth at that energy range. Clear k-edges at the measured energy range make the differentiation of materials much easier for the used material decomposition. The attenuation values in figure \ref{fig:Se_att_coeffs} are based on the attenuation value tables on the NIST (National Institute of Standards and Technology) X-ray attenuation online simulator \cite{saloman1988x}.  

\begin{figure}[H]
    \centering
    \includegraphics[scale = 0.5]{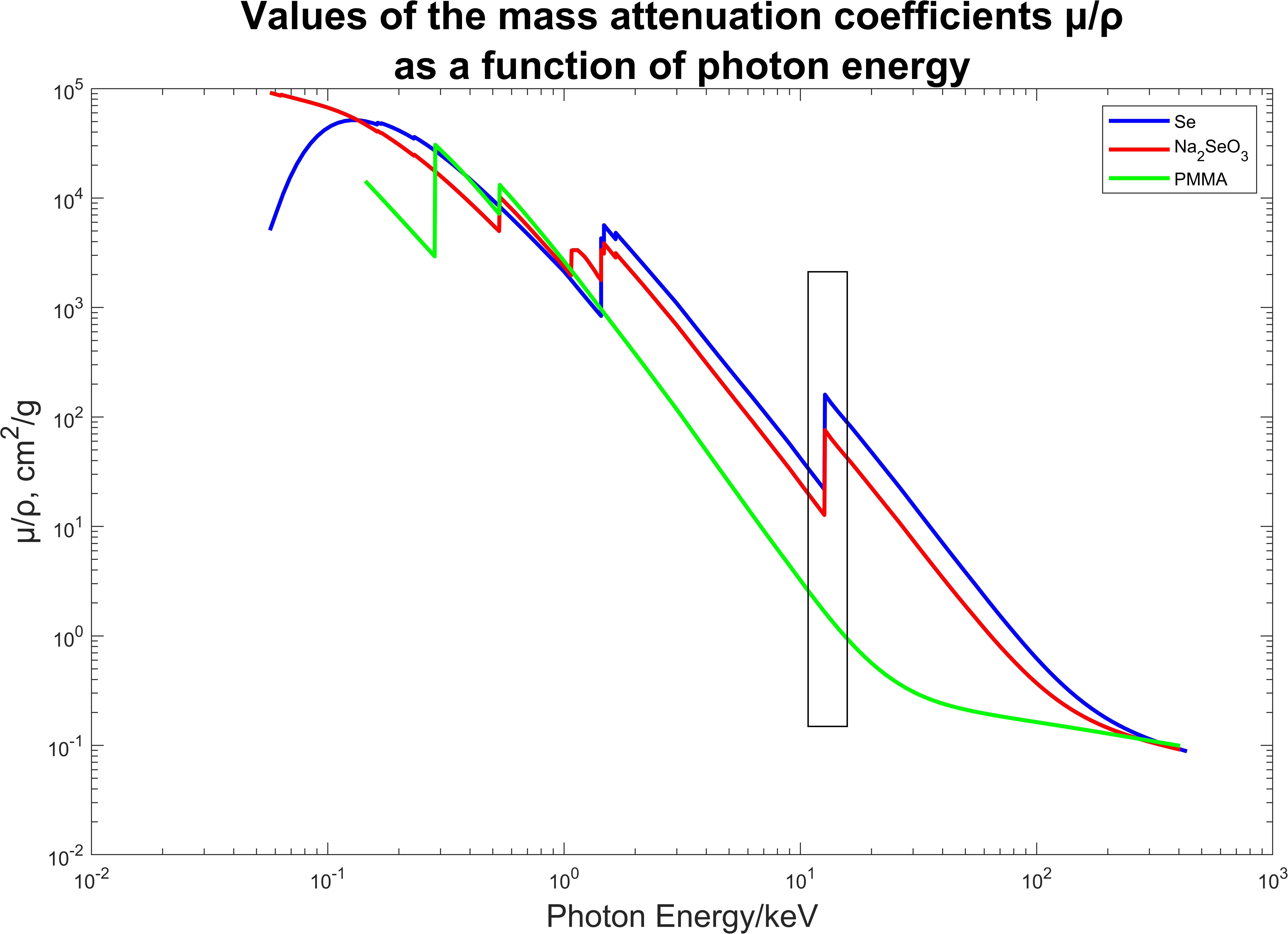}
    \caption{The mass attenuation coefficient, $\frac{\mu}{\rho}$, as a function of photon energy, for elemental Selenium Se, Na$_2$O$_3$ and PMMA. In the plot, the K-edges of the Se and Na$_2$O$_3$ are clearly visible, which makes the decomposing of the materials easier with the IP method. In the energy window that we use (12,54-12,80 keV), which is marked with a black box in the figure, the attenuation values of PMMA are smooth, which means that it appears similar in every energy images. So we can subtract it out from all the different energy sinograms, by subtracting low energy sinogram from all the others.}
    \label{fig:Se_att_coeffs}
\end{figure}

\section{Results}

In this section we show the results with the simulations and finally with real data. The section also includes a brief demonstration of the effect of the new regularization term. We have measured the quality of the simulations with some error measures, which will be explained below.

\subsection{Error measures of the simulations}

We were avoiding the inverse crime in our simulations by adding some noise and modelling error into our data. The noise level in the simulations was selected to be one percent, and we added modelling error by rotating the phantom 45 degrees and interpolating the difference. 

We calculated quality measures for the simulations by comparing the original phantoms with the resulting reconstructions. The classical errors were calculated correspondingly than in \cite{gondzio2022material}. We defined the $L_2$-error 
\begin{equation}
L_2\text{-error} = \frac{\text{norm}(\text{phantom}(:)-\text{reconstruction}(:))}{\text{norm(phantom(:))}}.
\label{L2}
\end{equation}
In addition to that we calculated the structural similarity index (SSIM) \cite{wang2004image} and Haar wavelet-based perceptual similarity index (HaarPSI) \cite{reisenhofer2018haar}. We have written down the mean values of the errors for all the materials, by summing up the numerical errors of the different materials and dividing by three. The same holds for the mean pixel error, which we have then listed as a percentage of correctly classified pixels (Corr.pix \%) in table \ref{tab: beta_table}.

\subsection{Effect of the new regularization term}

In this section we show how regularization effects to our reconstruction results with IP method, when using simulations. We use digital phantoms, because the point is just to demonstrate the effect of the inner product regularization parameter.

The coefficients which we used for the simulated materials with different energies are listed in the table \ref{tab:simulation_att_coefs}. 

\begin{table}[H]
    \centering
    \begin{tabular}{|c|c|c|c|}
    \hline
        Energy & M1 & M2 & M3 \\ \hline
        E1 & 22.73 & 8.56 & 3.51\\
        E2 & 5.95 & 12.32 & 10.88\\
        E3 & 7.81 & 3.51 & 27.77\\
        \hline
    \end{tabular}
    \caption{Attenuation coefficients $c_{11}$, $c_{12}$, ... of the demonstration materials, with increasing photon energies E1, E2 and E3.}
    \label{tab:simulation_att_coefs}
\end{table}

We chose resolution $N=128$ for the reconstructions. The number of used projection images was sparse in these simulations: 15 equispaced angles.

After choosing the most optimal $\alpha$ value, (here $\alpha$ = 20000) we start to increase the $\beta$ parameter from zero, to see how it affects the result. See figure \ref{fig:beta_demo}. The errors of the reconstructions, compared with the ground truth are listed in the table \ref{tab: beta_table}. We can see from the results that increasing the $\beta$ term reduces the error and improves the quality of the reconstructions. The range of $\beta$ is limited to be smaller than the $\alpha$ parameter to ensure that the problem stays convex. 

\begin{table}[H]
    \centering
    \begin{tabular}{|c|c|c|c|c|c|c|}
    \hline
         Mean errors: & L2   & SSIM & haarPSI & Corr.pix \%& $\alpha$ & $\beta$  \\ \hline
         First row:   & 0.60  & 0.25     & 0.14  &96   & 20000 & 0 \\
         Second row:  & 0.56  & 0.26     & 0.17  &96.4 & 20000 & 8000 \\
         Third row:   & 0.55  & 0.39     & 0.19  &96.3 & 20000 & 16000\\
         Fourth row:  & 0.59  & 0.41     & 0.21  &95.6& 20000 & 19500\\
         Ground truth:& 0     & 1        & 1     &100 &     &  \\\hline
    \end{tabular}
    \caption{Error table for different $\beta$ values. We can see that the $L_2$-errors\ref{L2} are getting smaller, structural similarity is increasing and also the HaarPSI index is improving, when we increase the $\beta$ regularization. We can not increase $\beta$ more than this, because the problem might turn to unstable. Unfortunately the $\beta$ regularization does not reduce the random noise spots around the actual target, which is natural consequence if one thinks how the term works. It penalises only the overlapping materials in different images and random errors in different images rarely overlap.}
    \label{tab: beta_table}
\end{table}

\begin{figure}[H]
    \centering
    \includegraphics[scale=0.8]{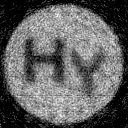}
        \includegraphics[scale=0.8]{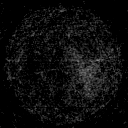}
            \includegraphics[scale=0.8]{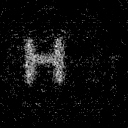}
              \includegraphics[scale=0.8]{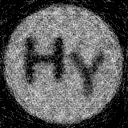}
        \includegraphics[scale=0.8]{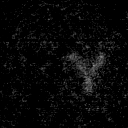}
            \includegraphics[scale=0.8]{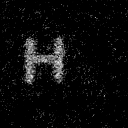}
               \includegraphics[scale=0.8]{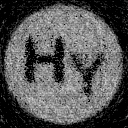}
        \includegraphics[scale=0.8]{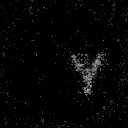}
            \includegraphics[scale=0.8]{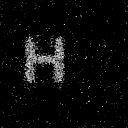}
                       \includegraphics[scale=0.8]{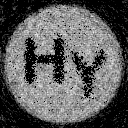}
        \includegraphics[scale=0.8]{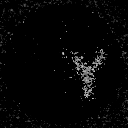}
            \includegraphics[scale=0.8]{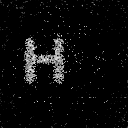}
                \includegraphics[scale=0.8]{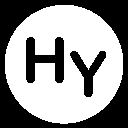}
        \includegraphics[scale=0.8]{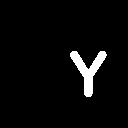}
            \includegraphics[scale=1.06]{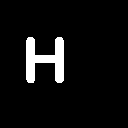}
    \caption{Beta demonstration. First row: $\beta$=0, second row: $\beta$=8000, third row: $\beta$=16000, fourth row: $\beta$ = 19500. We may notice from the figures that when we increase the $\beta$ regularization, the materials sharpen their shape in their own images, even though the effect is not very strong.}
    \label{fig:beta_demo}
\end{figure}

\begin{figure}[H]
    \centering
        \includegraphics[scale=0.6]{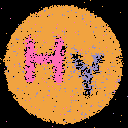}
        \includegraphics[scale=0.6]{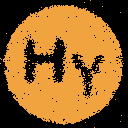}
        \includegraphics[scale=0.6]{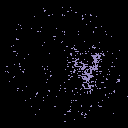}
        \includegraphics[scale=0.6]{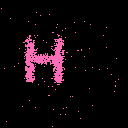}
        \includegraphics[scale=0.6]{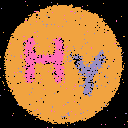}
        \includegraphics[scale=0.6]{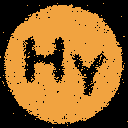}
        \includegraphics[scale=0.6]{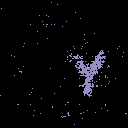}
        \includegraphics[scale=0.6]{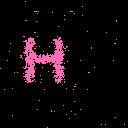}
        \includegraphics[scale=0.6]{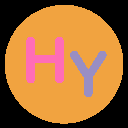}
        \includegraphics[scale=0.6]{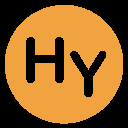}
        \includegraphics[scale=0.6]{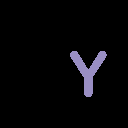}
        \includegraphics[scale=0.6]{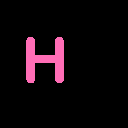}
    \caption{Beta demonstration. First row: $\beta$=0, second row: $\beta$=8000, third row: Ground truth. We may notice from the figures that when we increase the $\beta$ regularization, the materials sharpen their shape in their own images, even though the effect is not very strong.}
    \label{fig:beta_demo}
\end{figure}

\subsection{Results with simulated data}\label{sec:simulresults}

We have simulated the real data case virtually by using the measured attenuation values of the selenium compounds \ref{tab:att_coef}. We have the error measures separately for each material in table \ref{tab:simulation_errors}.

\begin{figure}[H]
    \centering
    \includegraphics[scale=0.86]{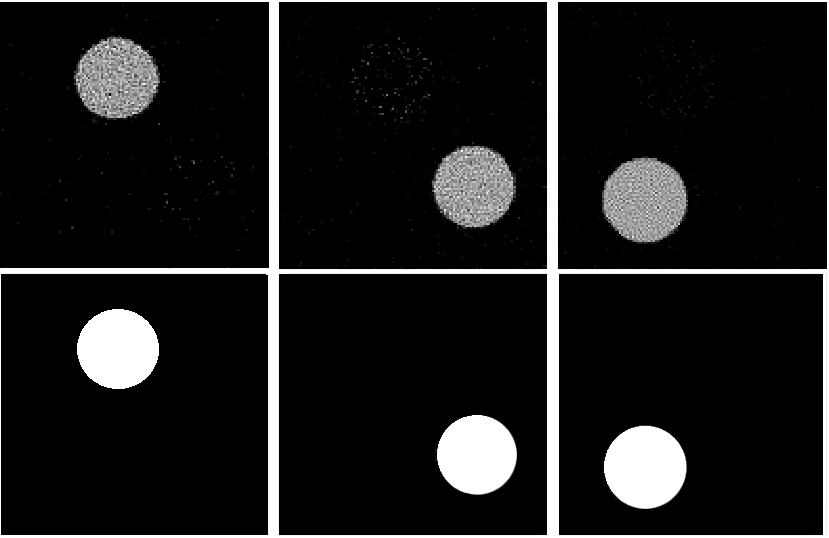}
    \caption{Results with simulated data. The first row is the simulated result for the different selenium samples. Here we have totally ignored the plastic container, which was present in the practical real world measurements. The second row shows the perfectly separated materials, which we have used as ground truths in the error calculations.}
    \label{fig:simulted_results}
\end{figure}

\begin{table}[H]
\caption{In this table, we have error calculations for the simulated phantom. We are aiming for as small as possible $L_2$-error\ref{L2}, but for SSIM and HPSI, the closer the value gets to one, the better, because it means that the images are more structurally similar or visually alike, which is the goal, when comparing the simulated results with the actual ground truth.}
\label{tab:error_table}

\centering
\begin{tabular}{|c|c|c|c|c|c|}
\hline
Material & Method & \phantom{mm}L2\phantom{mm}   & SSIM & HPSI\\
\hline
Se\hspace{0.5cm} & IP & 0.27 & 0.87 & 0.42\\
SeO$_3$ & IP          & 0.25 & 0.81 & 0.47\\
SeO$_4$ & IP          & 0.20 & 0.89 & 0.54\\
\hline
\end{tabular}
\label{tab:simulation_errors}
\end{table}

\subsection{Results with experimental X-ray data}

In this section, we show the reconstruction results with the standard filtered back projection (FBP) method and with the IP-method, see figure \ref{fig:real_data_results}. 

The number of used projection images is 200 with the real data.
We used the X-ray attenuation coefficients determined in \cite{honkanen2023monochromatic} in our system matrix, in place of the $c_{ij}$ coefficients. The X-ray attenuation coefficients used with different X-ray energies are listed in Table \ref{tab:att_coef}.

  \begin{table}[H]
      \centering
      \begin{tabular}{|c|c|c|c|}
      \hline
           Energy [keV] & Se     & SeO$_3$ & SeO$_4$ \\ \hline 
           12.658 & 8.4734 & 2.4380  & 3.5091\\
           12.662 & 7.9902 & 11.859 & 10.885\\
           12.685 & 7.4175 & 8.0033  & 27.777\\
           \hline
      \end{tabular}
      \caption{The X-ray attenuation coefficients for the different materials at different energies. The coefficients were determined in the article of Honkanen \& al.\cite{honkanen2023monochromatic}.}
      \label{tab:att_coef}
  \end{table}

Our simulation studies reported in Section \ref{sec:simulresults} suggest that $\alpha=20000$ and $\beta=19500$ are suitable choices for the regularization parameters. 


The reconstruction images in figure \ref{fig:real_data_results} have been black point and white point corrected for better visual impression. In these reconstructions, we have used all the projection images taken, so we have used data from 200 different projection angles. The first row of the figure shows FBP reconstructions made from the sinograms, which were measured with energies 12,658 keV, 12,662 keV and 12,685 keV. The end result is quite noisy and there are only small differences between the reconstructions. We have the PMMA phantom as a fourth material on the first row (and there was also some potassium as fulfilling material). The attenuation coefficients of these "extra" materials are quite same with every measured energy (they have no K-edges, see figure \ref{fig:Se_att_coeffs}). This for, we subtracted one low energy sinogram 12,645 keV (which was below all the K-edges) from the other sinograms before reconstruction. The result can be seen on the second and third row in the figure, where we have got rid of the phantom matrix. We have the FBP reconstructions on the second row and the IP reconstructions of the different materials on the third row.

\begin{figure}[H]
\begin{picture}(320,340) 
\put(0,380){FBP with energy 1}
\put(0,115){Material 1}
\put(0,265){\includegraphics[scale=0.4]{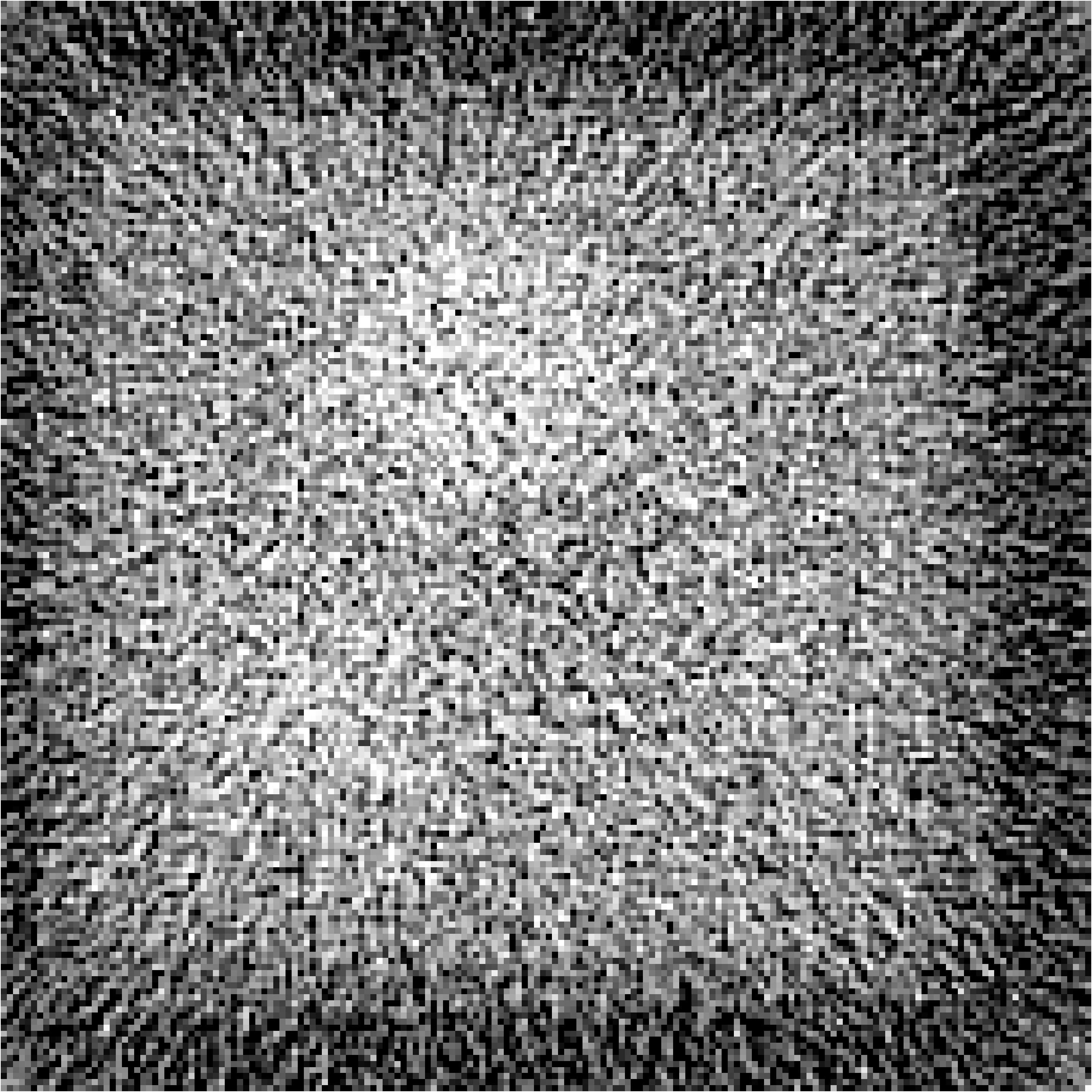}}
\put(0,155){\includegraphics[scale=0.4]{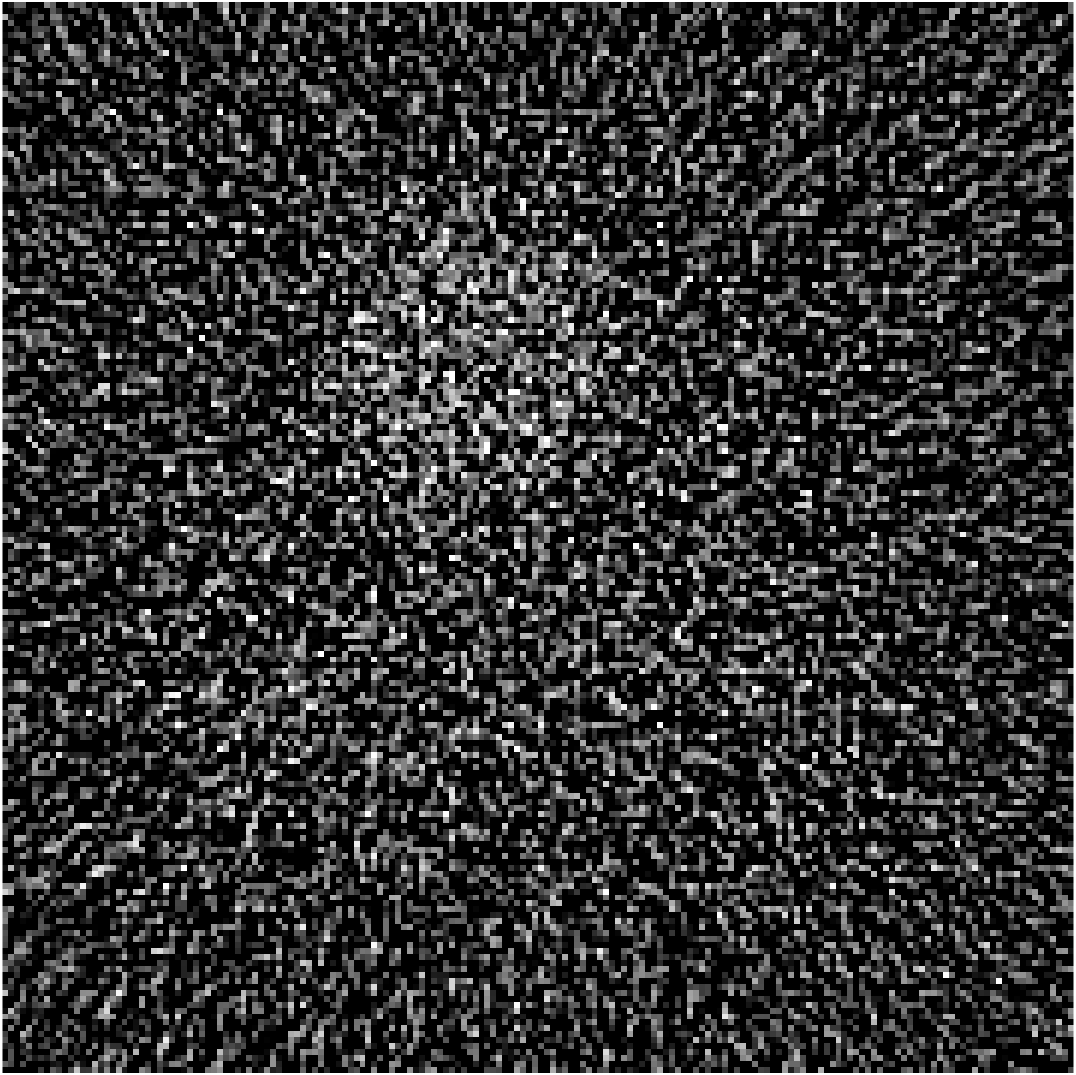}}
\put(0,0){\includegraphics[scale=0.4]{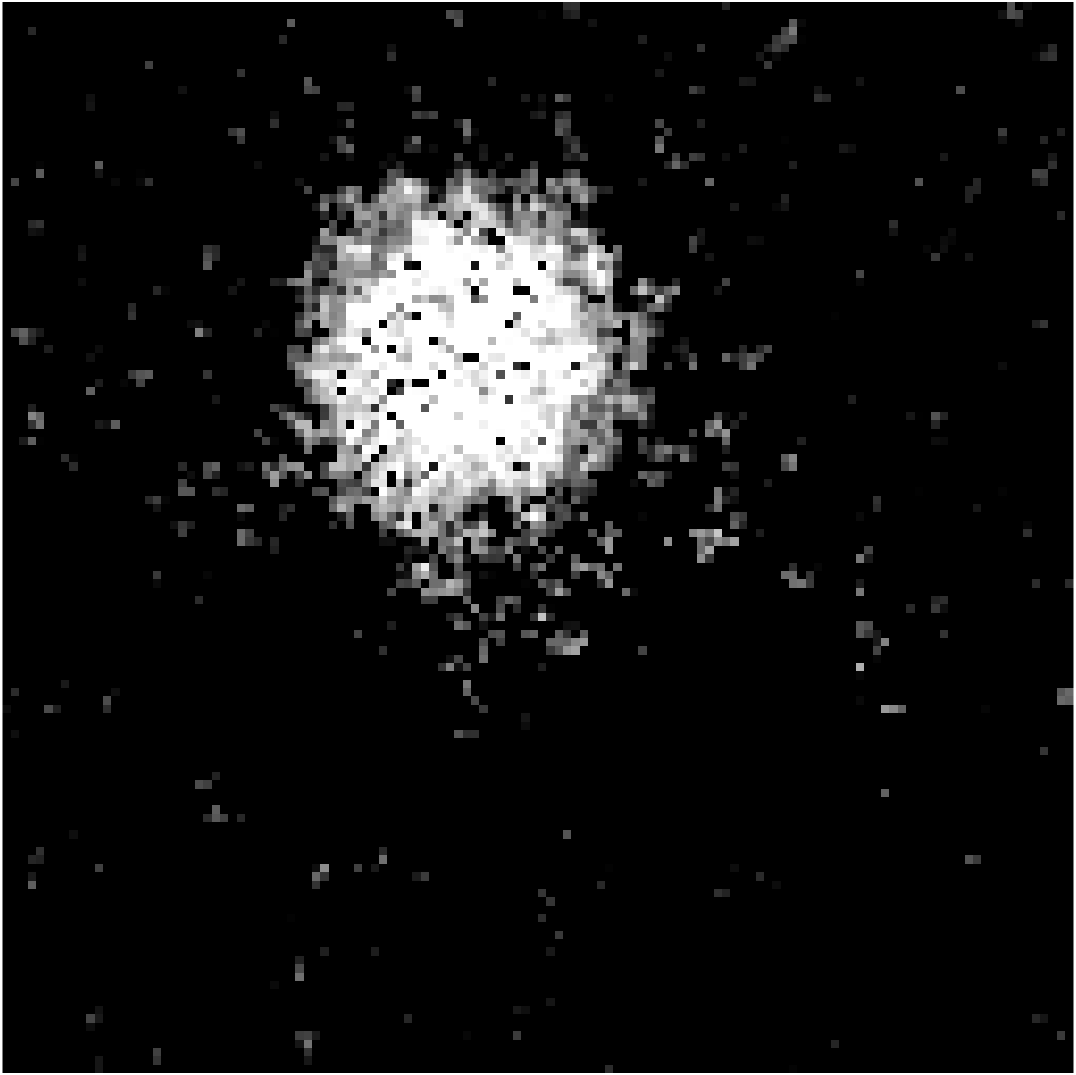}}
\put(110,380){FBP with energy 2}
\put(110,115){Material 2}
\put(110,265){\includegraphics[scale=0.4]{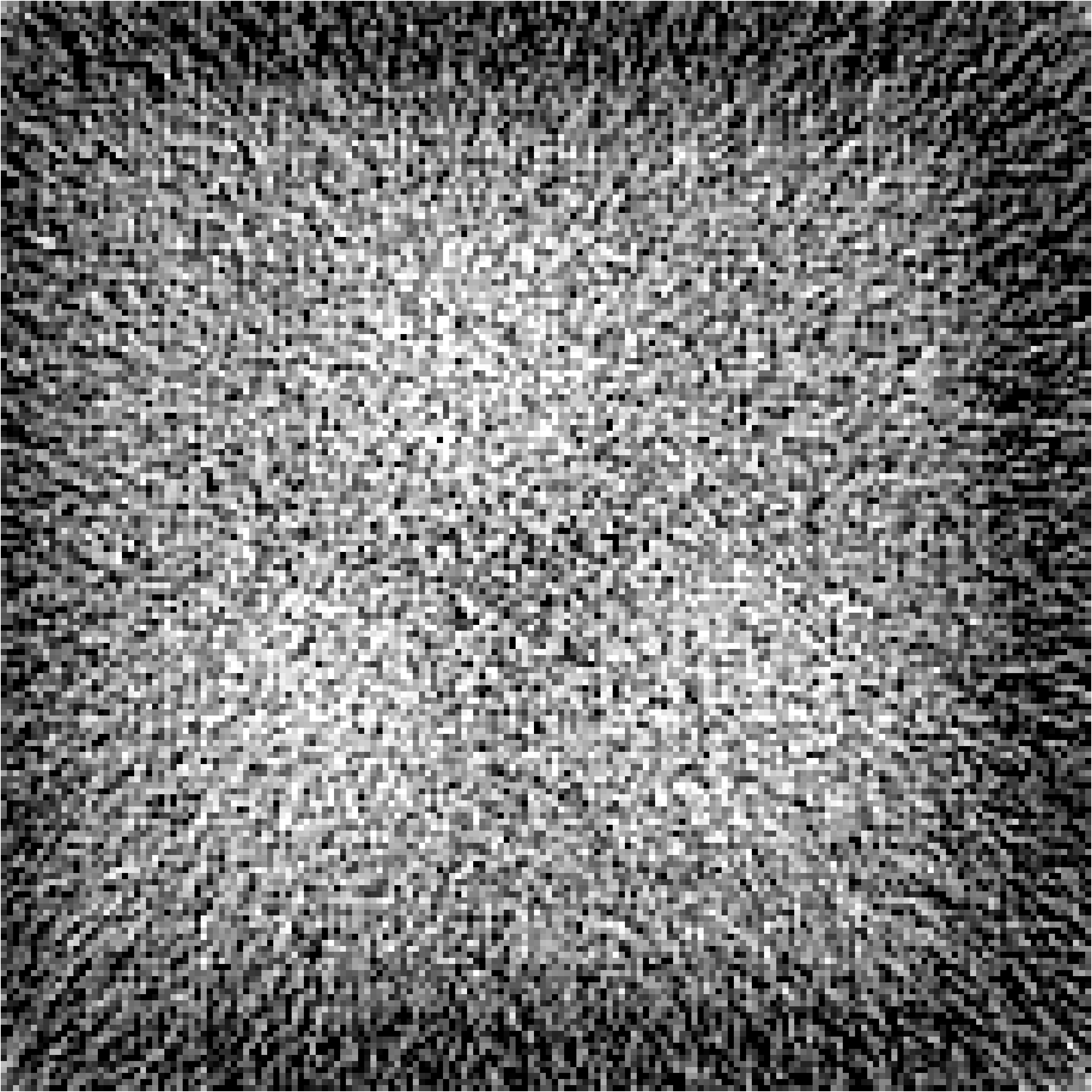}}
\put(110,155){\includegraphics[scale=0.4]{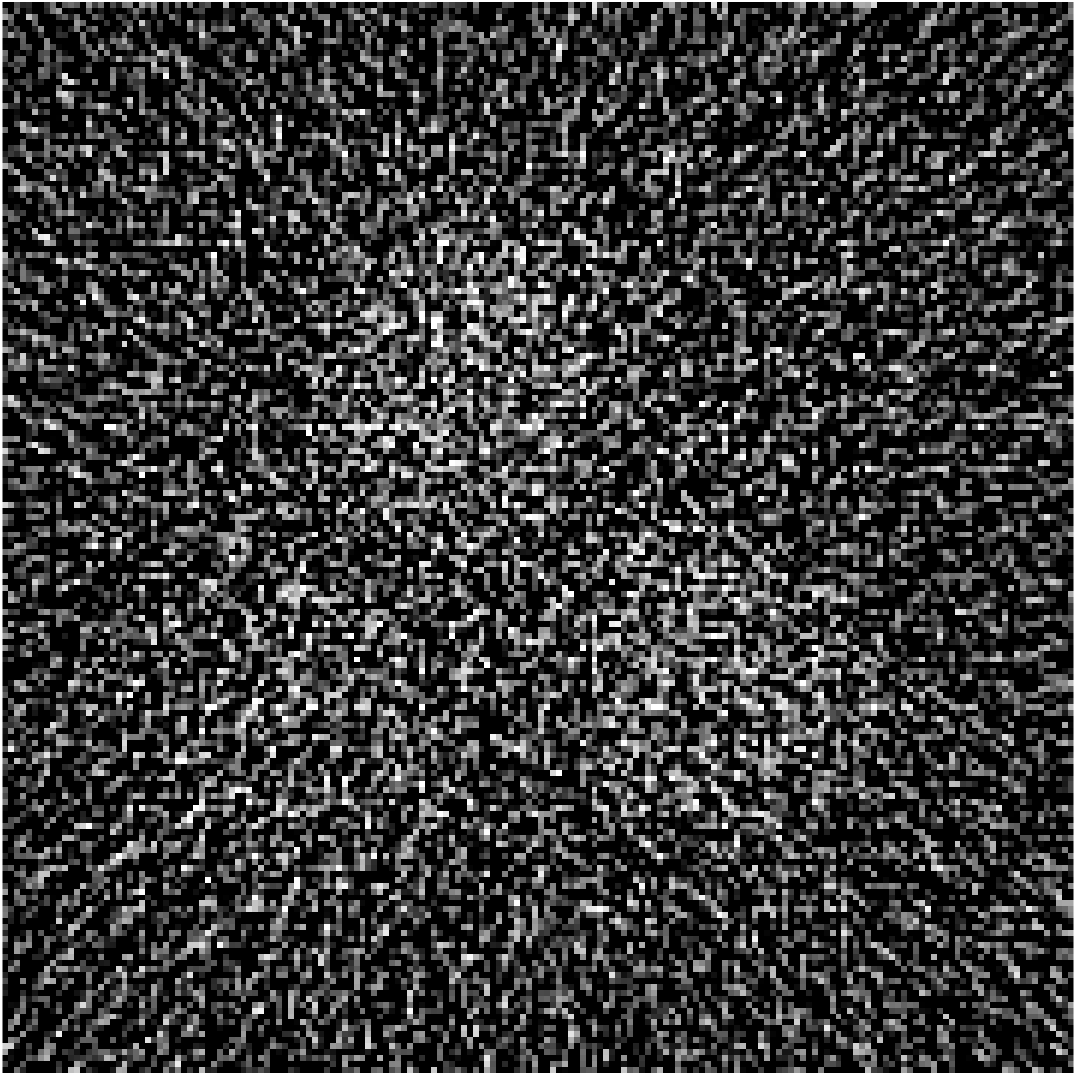}}
\put(110,0){\includegraphics[scale=0.4]{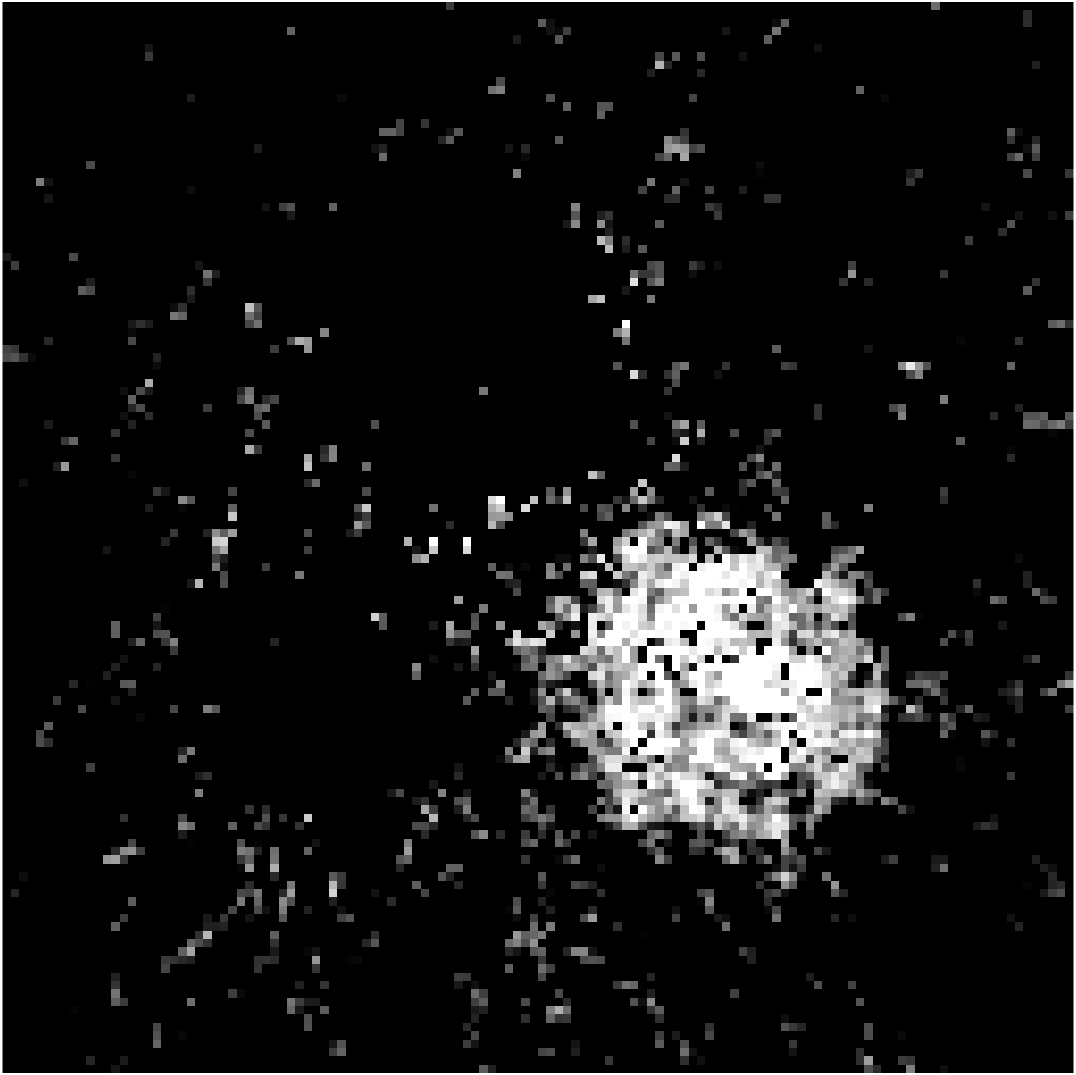}}
\put(220,380){FBP with energy 3}
\put(220,115){Material 3}
\put(220,265){\includegraphics[scale=0.4]{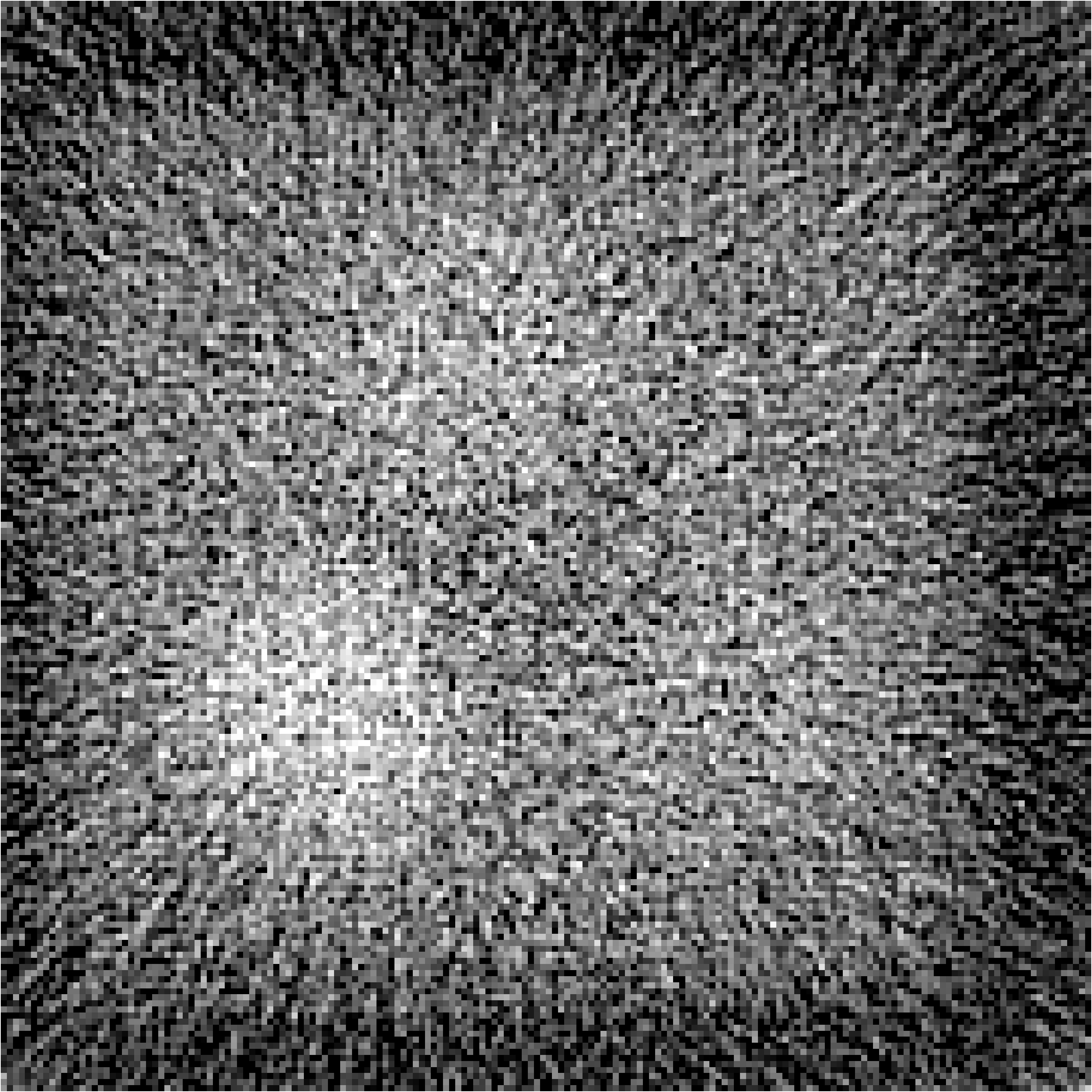}}
\put(220,155){\includegraphics[scale=0.4]{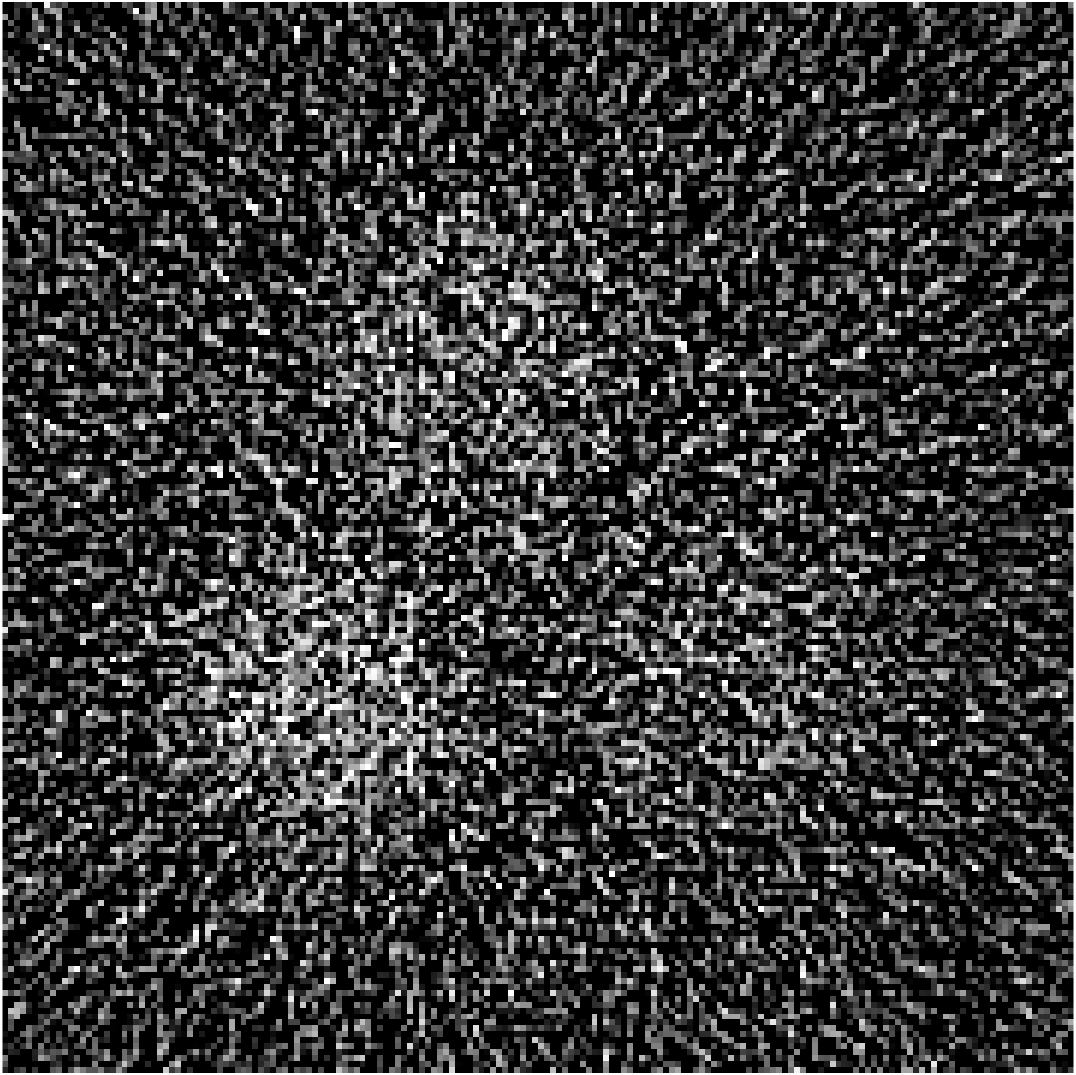}}
\put(220,0){\includegraphics[scale=0.4]{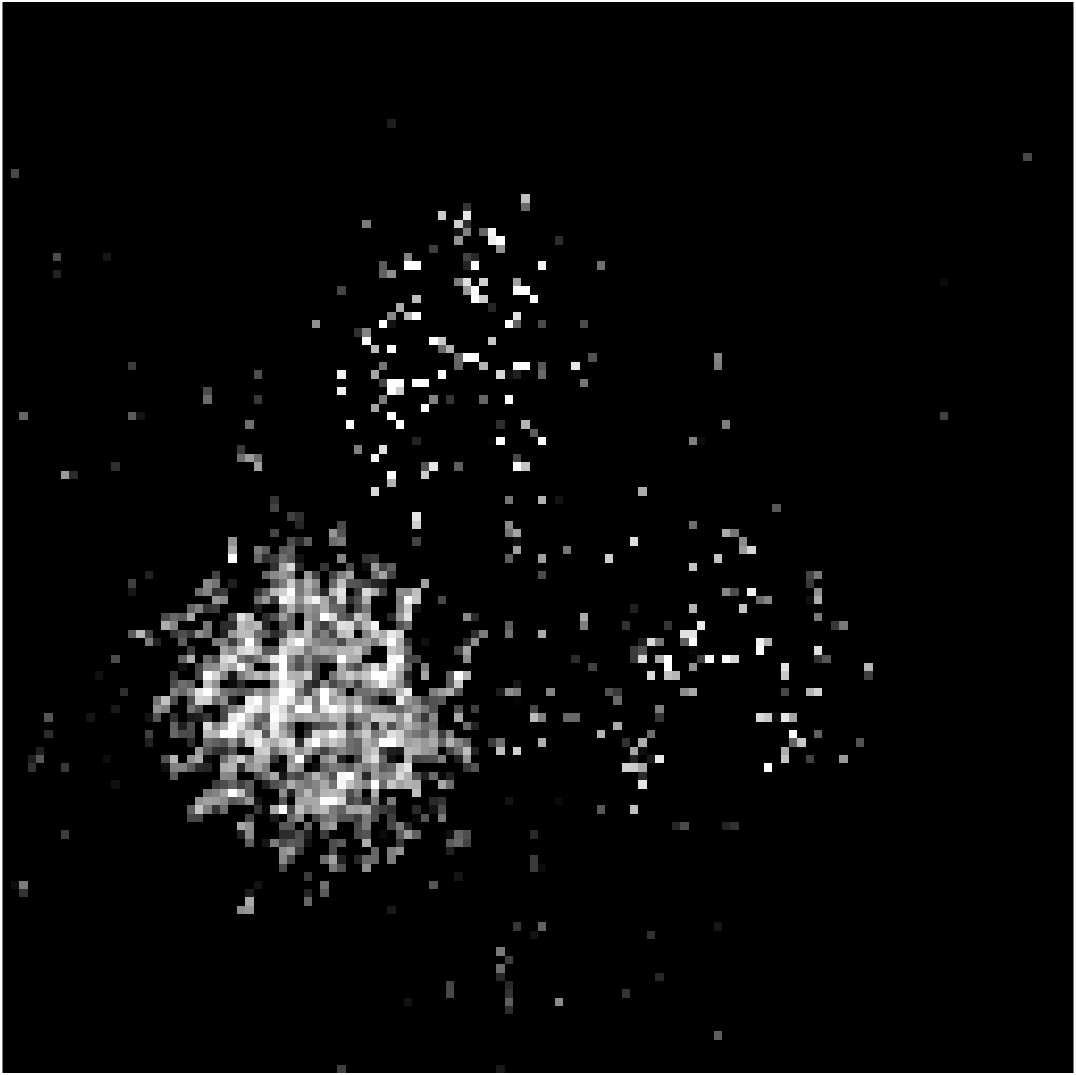}}
\put(-20,315){(a)}
\put(-20,205){(b)}
\put(-20,50){(c)}
 \end{picture}       
    \caption{Reconstruction results. Row (a) represents the FBP reconstructions, with the phantom matrix. Row (b) shows the result with the FBP method, when we have subtracted one measurement (taken below all the K-edges) from all sinograms before reconstruction. This subtraction removes the PMMA phantom matrix from the reconstructions. Row (b): IP reconstructions (also with subtracted phantom matrix). }
    \label{fig:real_data_results}
\end{figure}
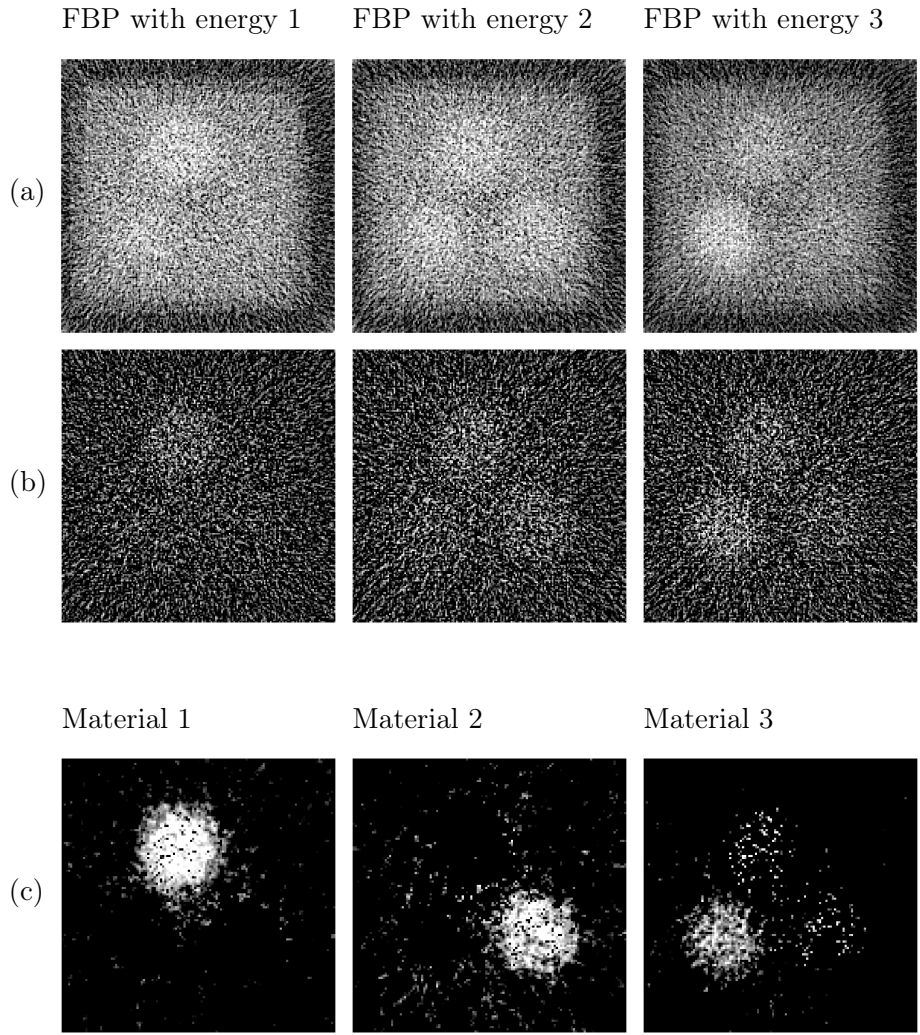

\section{Conclusions}

The Inner Product reconstruction approach seems to be capable for material decomposition in the case of real, monochromatic X-ray data. The reconstruction results show clearly where the different materials are located in the sample, even though there is some noticeable dotted noise around the samples. The theory of this method has been shown to work for multiple materials and X-ray energies and in this article we demonstrated the working in the case of three energies and three materials. 
As a drawback, we can notice that there must be clear difference in the attenuation coefficients of the materials with different energies for the proper working of the method. This is achieved using monochromatic X-ray beam and by selecting the used X-ray energies wisely just below and over the K-edge peak of the materials to be decomposed. When we keep these limitations in mind, the method works as a reliable and simple tool for material separation and identification.

\bibliographystyle{plain}
\bibliography{mybibliography}

\begin{thebibliography}{10}

\bibitem{batenburg20093d}
Kees~Joost Batenburg, Sara Bals, J~Sijbers, C~K{\"u}bel, PA~Midgley,
  JC~Hernandez, U~Kaiser, ER~Encina, EA~Coronado, and G~Van~Tendeloo.
\newblock 3d imaging of nanomaterials by discrete tomography.
\newblock {\em Ultramicroscopy}, 109(6):730--740, 2009.

\bibitem{batenburg2011dart}
Kees~Joost Batenburg and Jan Sijbers.
\newblock Dart: a practical reconstruction algorithm for discrete tomography.
\newblock {\em IEEE Transactions on Image Processing}, 20(9):2542--2553, 2011.

\bibitem{baumann2007discrete}
Joachim Baumann, Zolt{\'a}n Kiss, Sven Krimmel, Attila Kuba, Antal Nagy, Lajos
  Rodek, Burkhard Schillinger, and J{\"u}rgen Stephan.
\newblock Discrete tomography methods for nondestructive testing.
\newblock In {\em Advances in discrete tomography and its applications}, pages
  303--331. Springer, 2007.

\bibitem{gondzio_25}
J.~Gondzio.
\newblock Interior point methods 25 years later.
\newblock {\em European Journal of Operational Research}, 218:587--601, 2012.

\bibitem{gondzio2022material}
Jacek Gondzio, Matti Lassas, Salla-Maaria Latva-{\"A}ij{\"o}, Samuli Siltanen,
  and Filippo Zanetti.
\newblock Material-separating regularizer for multi-energy x-ray tomography.
\newblock {\em Inverse Problems}, 38(2):025013, 2022.

\bibitem{herman2008advances}
Gabor~T Herman and Attila Kuba.
\newblock {\em Advances in discrete tomography and its applications}.
\newblock Springer Science \& Business Media, 2008.

\bibitem{herman2012discrete}
Gabor~T Herman and Attila Kuba.
\newblock {\em Discrete tomography: Foundations, algorithms, and applications}.
\newblock Springer Science \& Business Media, 2012.

\bibitem{honkanen2023monochromatic}
Ari-Pekka Honkanen and Simo Huotari.
\newblock Monochromatic computed tomography using laboratory-scale setup.
\newblock {\em Scientific Reports}, 13(1):363, 2023.

\bibitem{long2014multi}
Yong Long and Jeffrey~A Fessler.
\newblock Multi-material decomposition using statistical image reconstruction
  for spectral ct.
\newblock {\em IEEE transactions on medical imaging}, 33(8):1614--1626, 2014.

\bibitem{reisenhofer2018haar}
Rafael Reisenhofer, Sebastian Bosse, Gitta Kutyniok, and Thomas Wiegand.
\newblock A haar wavelet-based perceptual similarity index for image quality
  assessment.
\newblock {\em Signal Processing: Image Communication}, 61:33--43, 2018.

\bibitem{saloman1988x}
EB~Saloman, JH~Hubbell, and JH~Scofield.
\newblock X-ray attenuation cross sections for energies 100 ev to 100 kev and
  elements z= 1 to z= 92.
\newblock {\em Atomic Data and Nuclear Data Tables}, 38(1):1--196, 1988.

\bibitem{solem2021material}
Rasmus Solem, Till Dreier, Isabel Gon{\c{c}}alves, and Martin Bech.
\newblock Material decomposition in low-energy micro-ct using a dual-threshold
  photon counting x-ray detector.
\newblock {\em Frontiers in Physics}, 9:673843, 2021.

\bibitem{wang2004image}
Zhou Wang, Alan~C Bovik, Hamid~R Sheikh, and Eero~P Simoncelli.
\newblock Image quality assessment: from error visibility to structural
  similarity.
\newblock {\em IEEE transactions on image processing}, 13(4):600--612, 2004.

\bibitem{wright}
S.~J. Wright.
\newblock {\em Primal-Dual Interior-Point Methods}.
\newblock SIAM, 1997.

\bibitem{wu2020dictionary}
Weiwen Wu, Haijun Yu, Peijun Chen, Fulin Luo, Fenglin Liu, Qian Wang, Yining
  Zhu, Yanbo Zhang, Jian Feng, and Hengyong Yu.
\newblock Dictionary learning based image-domain material decomposition for
  spectral ct.
\newblock {\em Physics in Medicine \& Biology}, 65(24):245006, 2020.

\end{thebibliography}
\end{document}